\documentclass[acmsmall,nonacm]{acmart}
\AtBeginDocument{%
  }

\setcopyright{acmlicensed}
\copyrightyear{2026}
\acmYear{2026}
\acmDOI{XXXXXXX.XXXXXXX}

\acmJournal{PACMPL}
\acmVolume{1}
\acmNumber{1}
\acmArticle{1}
\acmMonth{1}

\usepackage{booktabs}
\usepackage{subcaption}

\usepackage{preamble}

\newcommand{\kw}[1]{\mathtt{#1}}

\newcommand{\rele}{\lesssim}

\newcommand{\lsrc}{\textnormal{$\lambda_{\lk{DS}}$}\xspace}

\newcommand{\ctr}{{\mathrm C}}
\newcommand{\fun}[1]{{\sf fix}(#1)}
\newcommand{\lets}{{\sf let}}
\newcommand{\ins}{{\sf in}}
\newcommand{\match}{{\sf match}}
\newcommand{\with}{{\sf with}}

\newcommand{\clossrc}[2]{{\sf Clos}(#1, #2)}
\newcommand{\clostrg}[2]{{\kw{Clos}}(#1, #2)}

\newcommand{\evalsrc}[4]{#1; #2 \tensor*{\Downarrow}{^{#3}_{\kw{src}}} #4}

\newcommand{\evalsrcw}[3]{#1; #2 \Downarrow_{\kw{src}} #3}

\newcommand{\evaltrg}[4]{#1; #2 \tensor*{\Downarrow}{^{#3}_{\kw{trg}}} #4}
\newcommand{\evaltrgw}[3]{#1; #2 \Downarrow_{\kw{trg}} #3}

\renewcommand{\vec}[1]{\overline{#1}}
\makeatletter
\DeclareRobustCommand\onedot{\futurelet\@let@token\@onedot}
\def\@onedot{\ifx\@let@token.\else.\null\fi\xspace}

\def\eg{\emph{e.g}\onedot}

\makeatother

\newcommand{\bnfalt}{{\bf \,\,\mid\,\,}}
\newcommand{\bnfdef}{{\bf ::=}}

\newcommand{\lk}[1]{%
  \texttt{%
    \fontdimen2\font=2pt   
    #1}}

\newcommand{\pri}[1]{.#1}

\newcommand{\constr}[1]{\mathrm{C}(#1)}
\newcommand{\constri}[2]{\mathrm{C}_{#1}(#2)}
\newcommand{\clos}[2]{\langle#1, #2\rangle}

\newcommand{\lanf}{\textnormal{$\lambda_{\lk{ANF}}$}\xspace}

\newcommand{\val}{\kw{Val}}

\newcommand{\var}{\kw{Var}}

\newcommand{\letin}[3]{\lk{let}~#1=#2~\lk{in}~#3}

\newcommand{\halt}[1]{\lk{ret}(#1)}
\newcommand{\letfunin}[3]{\lk{fun}~#1=#2~\lk{in}~#3}
\newcommand{\letfun}[2]{\lk{fun}~#1=#2}

\newcommand{\casew}[2]{\lk{case}~#1~\lk{of}~[\mathrm{C}_i~\rightarrow~#2_i]_{i \in I}}

\newcommand{\fv}[1]{{\sf fv}(#1)}

\newcommand{\len}[1]{{\sf len}(#1)}

\newcommand{\True}{\lk{True}}
\newcommand{\False}{\lk{False}}

\newcommand{\defeq}{\overset{\mathrm{def}}{=}}

\newcommand{\res}[1]{{\sf Res}(#1)}
\newcommand{\oot}{{\sf OOT}}

\NewDocumentCommand\DownArrow{O{2.0ex} O{black}}{%
   \mathbin{\tikz[baseline] \draw [<-, line width=0.5pt, #2] (0,-0.07cm) -- ++(0,#1);}
}

\newcommand{\comp}{\circ}

\newcommand{\ctx}{\mathcal{E}}

\newcommand{\obs}{\mathrel{\approx}}

\DeclareMathAlphabet\mathbfcal{OMS}{cmsy}{b}{n}

\newcommand{\exprelww}[3]{\mathbfcal{E}^{#1} #2~#3 }

\newcommand{\valrelw}[3]{\mathbfcal{V}^{#1} (#2, #3) }

\newcommand{\resrelw}[3]{\mathbfcal{R}^{#1} (#2, #3) }

\newcommand{\src}{\mathrm{src}}
\newcommand{\trg}{\mathrm{trg}}

\newcommand{\anfrel}[6]{#1; #3; #2 \vartriangleright #4; #5; #6}
\newcommand{\anfexps}[6]{#1; #3; #2 \vartriangleright_{\sf e} #4; #5; #6}

\begin{document}

\title{Machine-Generated, Machine-Checked Proofs for a Verified Compiler (Experience Report)}

\author{Zoe Paraskevopoulou}
\affiliation{
  \institution{National Technical University of Athens}
  \country{Greece}
}
\email{zoepar@softlab.ntua.gr}
\begin{abstract}
  We report on using an agentic coding assistant (Claude Code, powered by Claude
  Opus~4.6) to mechanize a substantial Rocq correctness proof from scratch, with
  human guidance but without human proof writing.
  The proof establishes semantic preservation for the administrative normal
  form (ANF) transformation in the CertiCoq verified compiler for Rocq.
  The closely related continuation-passing style (CPS) transformation in
  CertiCoq was previously proved correct by human experts over several
  months~\cite{CPS:Paraskevopoulou21:compilingwithcontinuations}. We use
  this proof as a template and instruct the LLM to adapt the proof
  technique to the ANF setting, which differs in important technical ways.
  The resulting ANF proof comprises approximately 7,800 lines of Rocq (larger
  than the 5,300-line CPS proof) and was developed in approximately 96~hours.
  We describe the proof technique and report on the experience of
  developing it with an LLM, discussing both the strengths and
  limitations of the approach and its implications for verified compiler
  construction.
\end{abstract}


\maketitle

\section{Introduction}
\label{sec:intro}

Recent advances in large language models (LLMs) have transformed software
development.
Agentic coding assistants such as GitHub Copilot~\cite{github-copilot} and
Claude Code~\cite{claude-code} are now routinely used to generate and refactor
code, and early reports suggest that LLMs can dramatically accelerate software
development~\cite{peng2023copilot}.
At the same time, correctness of LLM-generated code remains a pressing issue.
Software bugs are pervasive even in carefully hand-written code; for
LLM-generated code that the user may not fully understand, the risks are
compounded.

This raises a question: can LLMs generate not only code, but also
\emph{machine-checked proofs of correctness} for that code?
If so, software could, in principle, come with formal guarantees certified by an
independent proof checker, guarantees that are unaffected by whether the code
was written by a human or an LLM.
This vision, \emph{machine-generated code with machine-generated,
machine-checked proofs}, represents a compelling direction for the future of
software engineering, and one that is becoming increasingly plausible.

This paper contributes a case study in this direction.
We investigate whether a state-of-the-art LLM, given a suitable template and
human guidance, can develop a substantial mechanized proof in the Rocq
proof assistant.
Our case study is grounded in a concrete, technically demanding compiler
verification task: proving the correctness of an ANF transformation in the
CertiCoq verified compiler~\cite{Coq:Anand17:certicoq}.

\paragraph{The case study.}
CertiCoq compiles Gallina (the specification language of the Rocq prover) to
C~\cite{savary19:codegen} and WebAssembly~\cite{certicoq-wasm}, via a pipeline
of transformations, a substantial portion of which has been formally
verified~\cite{Proofs:Paraskevopoulou19:compositionalopt,
Proofs:Paraskevopoulou19:closureconversion, Proofs:Belager17:shrinkfast}.
The pipeline employs MetaCoq~\cite{Coq:Sozeau19:themetacoq} as a front end to
extract an untyped intermediate representation from Gallina, which is then
translated into a functional intermediate language via either a
continuation-passing style (CPS) transformation or an administrative normal
form (ANF) transformation.
The correctness of the CPS transformation was established by
\citeN{CPS:Paraskevopoulou21:compilingwithcontinuations} using a novel proof
technique that combines syntactic simulation arguments with step-indexed logical
relations~\cite{Logrel:Ahmed06:step-indexsyntactic,Logrel:Appel01:anindexed}.
The proof, mechanized in Rocq, is approximately 4,200 lines and was developed
over several months by human experts.
The application of the same proof technique to the ANF transformation was left as
future work.

In this paper, we report on an experiment in which we used Claude Code (an
agentic LLM tool for software development) to develop the ANF correctness proof,
with the authors providing high-level guidance, help with reasoning using
natural language explanations (which was required in surprisingly few cases),
and and reviewing the generated proofs.
The authors did not write Rocq proofs directly and the proof was conducted from
scratch.
The experiment was largely successful.
The LLM produced approximately 6,000 lines of Rocq proof, covering the central
semantic preservation theorem for terminating programs, several auxiliary
lemmas, and various supporting definitions.
The LLM correctly identified that the technique that is used to show that the
CPS transformation preserves diverging behaviors does not generalize to the ANF
setting, proving a counterexample.
The proof follows the same high-level structure as the CPS proof, but, despite
the structural similarity, there are important technical differences in the
reasoning due to the different nature of the ANF transformation.
As an empirical evidence of complexity relative to the CPS proof, we report that
an author of this paper has previously attempted to develop the ANF proof
manually but set it aside due to complexity and time constraints.
%
%
The progress of the LLM-mechanized proof was documented with frequent commits to
a public repository that span the course of approximately 96 hours.

\paragraph{Contributions.}
This paper makes the following contributions:
\begin{itemize}
\item We present a mechanized correctness proof of the ANF transformation of the
  CertiCoq compiler (\cref{sec:anf-proof}), following the proof technique of the
  CPS correctness proof
  of~\citeN{CPS:Paraskevopoulou21:compilingwithcontinuations}.
  The correctness statement covers semantic preservation for terminating
  source programs.
%

\item We report on the experience of using an LLM (Claude Code,
  powered by Claude Opus 4.6) to develop this proof
  (\cref{sec:experience}).
  We describe the experimental setup, the interaction process, the
  challenges encountered, and the lessons learned.

\item We discuss the broader implications of LLM-assisted proof
  development for programming languages research and verified compiler
  construction (\cref{sec:discussion}).
\end{itemize}

\paragraph{Paper overview.}
In \cref{sec:technique}, we overview the proof technique used for the CPS
transformation.
\Cref{sec:tech} provides the technical setup for the proof, including the source
and target languages, their semantics, and the logical relation.
In \cref{sec:anf-proof}, we describe the ANF transformation and its correctness proof.
\Cref{sec:experience} reports on the experience of developing the proof with an LLM,
and \cref{sec:discussion} discusses implications and future directions.
We cover related work in \cref{sec:related} and conclude in \cref{sec:conclusion}.
\Cref{app:statistics} presents detailed proof development statistics and
\cref{app:prompts} provides a curated selection of user prompts from the
conversation transcripts.

\section{The CPS Proof Technique}\label{sec:technique}

Proving a compiler transformation correct amounts to proving that the observable
behaviors of the compiled program are a subset of the observable behaviors of
the source program.
For deterministic languages, this can be expressed as forward
simulation~\cite{leroy09}: whenever the source program takes a step $e \to e'$,
the compiled program takes one or more steps to reach the translation of $e'$,
that is,  $[ e ] \to^+ [ e' ]$.
There are many transformation that satisfy this simple commuting diagram, but
the CPS transformation is not one of them due to well-studied technical
obstacles~\cite{CPS:Dargaye07:mechanizedverification}.

Let $[e]_k$ denote the CPS translation of an expression $e$ with current
continuation $k$, and $\bar{v}$ the CPS translation of a value.
The application case of a CPS transformation can be defined as
$[v_1\;e_2]_k \;=\; (\lambda m.\;[e_2]_{\lambda n.\,m\,k\,n})\;\bar{v}_1$.
The outer $(\lambda m.\;\ldots)\;\bar{v}_1$ is an
\emph{administrative redex}: a $\beta$-redex that the transformation
introduces with no source counterpart.

Suppose $v_1\;e_2 \to v_1\;e_2'$ where $e_2 \to e_2'$.
A naive simulation tries to show $[v_1\;e_2]_k \to^+ [v_1\;e_2']_k$.
Contracting the administrative redex and applying the induction
hypothesis on $e_2 \to e_2'$ brings the target to
$[e_2']_{\lambda n.\,\bar{v}_1\,k\,n}$.
But the goal $[v_1\;e_2']_k = (\lambda m.\;[e_2']_{\lambda n.\,m\,k\,n})\;\bar{v}_1$
is \emph{not} this term: it is a fresh administrative redex whose
contractum is exactly what we reached.
We would need to step \emph{backwards} to close the diagram, which is
impossible.
The naive simulation fails.

An orthogonal obstacle arises from fresh binder names.
As ~\citeN{CPS:Paraskevopoulou21:compilingwithcontinuations} explain, when $ e
\to e' $ is a reduction step in the source, $[ e ]_k$ may reduce to a term that
is $\alpha$-equivalent to $[ e' ]_k$ but not syntactically equal, because two
invocations of the CPS transformation may use different fresh-name counters.
The simulation diagram must then be extended with $\alpha$-equivalence, which
is possible but tedious and brittle.


The key insight of \citeN{CPS:Paraskevopoulou21:compilingwithcontinuations} is
to replace the reduction relation at the target level with \emph{logical
relatedness}.
Rather than requiring that $[e]_k$ reduces to exactly $[e']_k$ after
each source step $e \to e'$, the proof only requires:
\[
  e \to e' \quad\Longrightarrow\quad
  [e]_k \;\rele\; [e']_k.
\]

where $\rele$ is a standard logical relation defined over terms which is
reflexive and transitive (we give a formal definition of the logical relation
used in this setting in \cref{sec:logrel}).

Returning to the example: instead of finding a reduction sequence from
$[v_1\;e_2]_k$ to $[v_1\;e_2']_k$, we show they are logically related.
The logical relation satisfies \emph{Reduce-App}: a $\beta$-redex is
logically related to its contractum (in either direction). In addition, 
$\alpha$-equivalent terms are logically related, so there is no need to 
reason explicitly about $\alpha$-equivalence.
Together with transitivity and the induction hypothesis, the chain
closes cleanly:
\begin{align*}
  [v_1\;e_2]_k
  &= (\lambda m.\;[e_2]_{\lambda n.\,m\,k\,n})\;\bar{v}_1 \\
  &\rele^i [e_2]_{\lambda n.\,\bar{v}_1\,k\,n}
    && \text{(Reduce-App)} \\
  &\rele^i [e_2']_{\lambda n.\,\bar{v}_1\,k\,n}
    && \text{(IH)} \\
  &\rele^i (\lambda m.\;[e_2']_{\lambda n.\,m\,k\,n})\;\bar{v}_1
    && \text{(Reduce-App)} \\
  &= [v_1\;e_2']_k
\end{align*}
Transitivity absorbs the administrative redexes on both sides.
There is no need to step backwards: the relation sees through the
administrative structure to the underlying semantics.

\section{Technical Setup}\label{sec:tech}
In this section, we provide the technical setup of our experiment, including the
source and target languages, their semantics, and the logical relation used in
the proof. These are all parts of the CertiCoq compiler.
For the full definitions of the semantics of the two languages, we refer the reader
to \citet{CPS:Paraskevopoulou21:compilingwithcontinuations}.

\subsection{Source Language}

The source language, \lsrc, is an untyped, call-by-value lambda calculus using
de~Bruijn indices. The grammar is shown in \cref{fig:src-grammar}.
Variables are natural numbers; the environment $\vec{v}$ is a
sequence of values, and the expression $n$ evaluates by looking up
the $n$-th entry.
The language is extended with let-bindings, fixpoints $\fun{e}$, constructor
application and pattern matching. 
Mutually recursive functions are not shown in the paper for simplicity but are
supported in the actual definitions.

\begin{figure}
\begin{subfigure}[t]{0.48\textwidth}
\input{src}
\caption{Source language \lsrc.}
\label{fig:src-grammar}\label{fig:src-grammar}
\end{subfigure}
\hfill
\begin{subfigure}[t]{0.48\textwidth}
\input{trg}
\caption{Target language \lanf.}
\label{fig:trg-grammar}\label{fig:trg-grammar}
\end{subfigure}
\caption{Language grammars.}
\label{fig:grammars}
\end{figure}

The semantics of \lsrc\ is given by a big-step evaluation relation written
 $\evalsrc{\vec{v}}{e}{f}{\res{v}}$. $\vec{v}$ is the evaluation environment,
 $e$ is the evaluated expression, $f$ is the fuel value which decreases with
 each step of evaluation, and $\res{v}$ is the result of the execution, which
 can be either a value $v$ or an out-of-time result $\oot$, indicating
 that the fuel has run out.

\subsection{Target Language}

The target language, \lanf, is a lambda calculus in administrative normal form,
with named variables.
In this representation, every intermediate result is bound to a name.
The grammar is shown in \Cref{fig:trg-grammar}.
The seven expression forms are: constructor allocation
$\letin{x}{\ctr(\vec{y})}{e}$, field projection
$\letin{x}{y\pri{i}}{e}$, function definition
$\letfunin{f~\vec{x}}{e_1}{e_2}$, non-tail call
$\letin{x}{f~\vec{y}}{e}$, case dispatch $\casew{y}{e}$,
tail call $f~\vec{x}$, and the return form $\halt{x}$.
Environments $\sigma$ are finite maps from variables to values;
values are constructor values $\ctr(\vec{v})$ or closures
$\clos{\sigma}{\letfun{f~\vec{x}}{e}}$.

The semantics of the target language is given by a big-step operational
semantics, written $\evaltrg{\sigma}{e}{f}{r}$, where $\sigma$ is the
evaluation environment, $e$ is the evaluated expression, $f$ is the fuel value,
and $r$ is the result of the execution, which can be either a value $v$ or
an out-of-time result $\oot$.

\paragraph{1-hole context.}
We define \emph{1-hole contexts} $\ctx$ as \lanf\ expression with a
hole $\bullet$ in tail-position:
\[
\begin{array}{r r l}
  \ctx & \bnfdef & \bullet \\
       & \bnfalt & \letin{x}{\ctr(\vec{y})}{\ctx} \\
       & \bnfalt & \letin{x}{y\pri{i}}{\ctx} \\
       & \bnfalt & \letin{x}{f~\vec{y}}{\ctx} \\
       & \bnfalt & \letfunin{f~\vec{x}}{e}{\ctx}
\end{array}
\]
We write $\ctx[\,e\,]$ for the expression obtained by replacing
$\bullet$ with $e$, and $\ctx_1 \comp \ctx_2$ for context
composition.
The ANF transformation produces pairs $(C, r)$---a context and a
result variable---such that the full target expression is
$C\,[\,\halt{r}\,]$.

\subsection{Logical Relation}
\label{sec:logrel}

The correctness proof uses step-indexed, untyped logical
relations.
This logical relation has been used to prove the correctness of all
\lanf-to-\lanf transformations in
CertiCoq~\cite{Proofs:Paraskevopoulou19:compositionalopt}, and in the correctess
proof of the CPS transformation (albeit in a different way).
The relation is defined mutually over values,
results, and expressions (\cref{fig:logrel}).

\begin{figure}
\input{logrel}
\caption{Logical relation over \lanf\ terms.}
\label{fig:logrel}
\end{figure}

Two values are related at step $k$, written $\valrelw{k}{v_1}{v_2}$,
as follows.
Two constructor values are related when their tags agree and their
fields are pointwise related.
Two closures $\clos{\sigma_1}{\letfun{f~\vec{x}}{e_1}}$ and
$\clos{\sigma_2}{\letfun{g~\vec{y}}{e_2}}$ are related when, for
every $i < k$ and every pair of argument vectors
$\valrelw{i}{\vec{v_1}}{\vec{v_2}}$, the two bodies yield
related expression configurations in their respective extended
environments.
The strict decrease from $k$ to $i < k$ at each function application
ensures that the definition is well-founded despite the absence of types.
The result relation $\resrelw{k}{r_1}{r_2}$ lifts value
relatedness: $\oot$ and $\oot$ are always related, $\res{v_1}$ and
$\res{v_2}$ are related when $v_1$ and $v_2$ are value-related, and
no other pairs are related.

The expression relation
$\exprelww{k}{(\sigma_1,e_1)}{(\sigma_2,e_2)}$ is the
heart of the definition: it asserts that for every evaluation of
$(\sigma_1, e_1)$ to a result $r_1$, there exists an evaluation of
$(\sigma_2, e_2)$ to a related result $r_2$.
In the full definition, the relation is additionally parameterized by
\emph{resource invariants} on evaluation costs, used to establish
\emph{divergence preservation} and \emph{asymptotic cost bounds}; we
refer
to~\cite{CPS:Paraskevopoulou21:compilingwithcontinuations,Proofs:Paraskevopoulou19:compositionalopt}
for details.

\section{The ANF Transformation and Its Correctness Proof}
\label{sec:anf-proof}
The ANF transformation proof is done in two layers: first, we define the ANF
transformation as a relation $\anfrel{S}{e}{\vec{x}}{C}{r}{S'}$ and
prove semantic preservation for this relation.
Then, we show that the executable ANF transformation, defined using a state
monad over a supply of fresh names, satisfies the relational specification.
This allows us to separate the monadic reasoning from the core simulation
argument, leading to a cleaner and more modular proof structure.
The proof of all \lanf transformations follows this pattern.

Unlike CPS, the ANF transformation does not introduce continuations.
Functions remain in direct style, and intermediate results are named
with let-bindings rather than being passed to continuations.
The ANF transformation is written $\anfrel{S}{e}{\vec{x}}{C}{r}{S'}$, where $S$
is a set from which fresh variables are drawn, $\vec{x}$ is a de Bruijn
environment mapping de Bruijn indices to variable names and $e$ is a source
expression. $C$ is a target 1-hole context and $r$: the context $C$ computes the
result of $e$ and binds it to $r$.
The full ANF expression corresponding to $e$ is $C[\,\halt{r}\,]$.
We explain some illustrative cases of the transformation below; the definition
is shown in \Cref{fig:anf-rel}.

\begin{figure}
\input{anf}
\caption{ANF transformation (relational definition, excerpt).}
\label{fig:anf-rel}
\end{figure}

\noindent\textbf{Variable.} A DeBruijn index $n$ is translated to the
corresponding variable in the environment $\vec{x}[n]$ and the empty context. 
%

\noindent\textbf{Lambda.}
A lambda $\fun{e}$ introduces a fresh function name $f$ and parameter
$x_1$, translates the body $e$ in environment $x_1 :: \vec{x}$ to
to obtain the context--variable pair $(C_1, r_1)$, and produces context
$\letfunin{f~x_1}{C_1[\,\halt{r_1}\,]}{\bullet}$ with result $f$.

\noindent\textbf{Application.}
An application $e_1\;e_2$ translates each subexpression, yielding
contexts $C_1, C_2$ and result variables $x_1, x_2$, and produces
context $C_1 \comp C_2 \comp \letin{r}{x_1~x_2}{\bullet}$ with $r$ 
holding the result.


%


\subsection{Proof Statements}
The main correctness theorem relates the source evaluation of $e$ to
the ANF-compiled target context $C$ via the logical relation.
The statement is adjusted to work with the big-step semantics
of CertiCoq.
It states that the source program $e$ evaluates to a value $v$ in environment
$\vec{v}$, and the ANF transformation produces context $C$ with result variable
$r$; then, for any continuation $e_k$, running $C$ followed by $e_k$ in the
target is logically related to running $e_k$ directly with $r$ bound to the
translated value $v'$.
We relate $v$ and $v'$ via the ANF value translation relation, written
$\lk{anf\_val\_rel}$. We lift this relation to environments with
$\lk{anf\_env\_rel}$. 



\begin{theorem}[ANF Correctness]\label{thm:anf-correct}
  Let $\vec{v}$ be a well-formed environment, $e$ a well-formed source
  expression, and $\evalsrcw{\vec{v}}{e}{v}$.
Suppose $\anfrel{S}{e}{\vec{x}}{C}{x}{S'}$.
Then for any target environment $\sigma$ such that $\lk{anf\_env\_rel}(\vec{x},
\vec{v}, \sigma)$, and target value such that $\lk{anf\_val\_rel}(v, v')$
we have 
      \[
        \forall~i,~\exprelww{i}
          {(e_k,\; \sigma[x \mapsto v'])}
          {(C[\,e_k\,],\; \sigma)}
      \]
where $e_k$ is any continuation expression such that $\fv{e_k} \cap ((S
\setminus S') \setminus \{x\}) = \emptyset$.


  
\end{theorem}

%
The universally quantified continuation $e_k$ is what makes the statement
compositional: it allows the theorem to be applied inductively in larger
contexts, as illustrated by the Let case below.

\paragraph{Example: the Let case.}
To illustrate the proof structure, consider the \textbf{Let} case.
Suppose $\evalsrcw{\vec{v}}{\lets~e_1~\ins~e_2}{r}$, which decomposes
as $\evalsrcw{\vec{v}}{e_1}{v_1}$ and $\evalsrcw{v_1 :: \vec{v}}{e_2}{r}$.
The ANF transformation yields contexts $C_1$, $C_2$ and result
variables $x_1$, $x_2$ such that $\anfrel{S}{e_1}{\vec{x}}{C_1}{x_1}{S_2}$
and $\anfrel{S_2}{e_2}{x_1 :: \vec{x}}{C_2}{x_2}{S'}$, with
output context $C_1 \comp C_2$ and result $x_2$.
We must show:
\[
  \exprelww{i}{(e_k,\; \sigma[x_2 \mapsto v_2'])}{((C_1 \comp C_2)[\,e_k\,],\; \sigma)}
\]
The induction hypotheses for the two sub-evaluations give:
\begin{align*}
  \exprelww{i}{(e_k,\; \sigma[x_1 \mapsto v_1'][x_2 \mapsto v_2'])}
    {(C_2[\,e_k\,],\; \sigma[x_1 \mapsto v_1'])}
    && \text{(IH for $e_2$, continuation $e_k$)} \\
    \exprelww{i}{(C_2[\,e_k\,],\; \sigma[x_1 \mapsto v_1'])}{(C_1[\,C_2[\,e_k\,]\,],\; \sigma)}
    && \text{(IH for $e_1$, continuation $C_2[\,e_k\,]$)}
\end{align*}
Transitivity closes the chain.

\paragraph{Top-level correctness.}
Following \citet{CPS:Paraskevopoulou21:compilingwithcontinuations}, from \cref{thm:anf-correct} we can derive a
whole-program correctness corollary.
We first define an observation relation $\obs$ between source and
target values:
\[\begin{array}{l}
  \ctr(\vec{v}) \obs \ctr(\vec{v}') \;\defeq\; \forall\,i,\; v_i \obs v_i' \\
  \clossrc{\vec{v}}{\fun{e}} \obs \clostrg{\sigma}{\letfun{f~x}{e'}} \;\defeq\; \True
\end{array}\]
Constructor values are related when they share the same outermost
constructor and have pairwise related arguments; closures are
trivially related.
We then define whole-program refinement $e_{\src} \sqsubseteq e_{\trg}$:
\[
  e_{\src} \sqsubseteq e_{\trg} \;\defeq\;
  \Bigl(\forall\,v,\;
    \evalsrcw{[]}{e_{\src}}{v} \Rightarrow
    \exists\,v',\;
    \evaltrgw{\emptyset}{e_{\trg}}{v'} \wedge v \obs v'
  \Bigr)
\]

\begin{corollary}[Whole-Program ANF Correctness]\label{cor:anf-whole} If $e$ is
  well-formed, then there exist $C$ and $x$ such that
  $\anfrel{S}{e}{[]}{C}{x}{S'}$, and
  \[
    e \;\sqsubseteq\; C[\,\halt{x}\,]
  \]
  where $C[\,\halt{x}\,]$ is evaluated from the empty target environment.
\end{corollary}

A similar corollary is derived for separate compilation, which we elide for brevity.

\subsection{Divergence Preservation}\label{sec:divergence-preservation} The
above theorem only establishes semantic preservation 
for terminating programs.
The CPS correctness proof, also establishes divergence preservation.
This is done by additionally proving than wen the source program runs out of
fuel, $\evalsrc{\vec{v}}{e_{\src}}{f_{\src}}{\oot}$, then so does the target,
$\evaltrg{\sigma}{C[\,e_k\,]}{f_{\trg}}{\oot}$, for some $f_{\trg} \geq
f_{\src}$.
This is a standard technique for proving divergence preservation with fuel-based
big-step semantics, used, among other by the CakeML
compiler~\cite{Cakeml:Owens17:veryfingefficient}.

However, this technique fails when the transformation can decrease the number of
steps taken by a program, as is the case for ANF.
Consider the source program that applies a variable to another $x_1~x_2$ and its
ANF translation $\letin{r}{x_1~x_2}{\halt{r}}$.
With the existing cost models of the source and target languages, the source
program evaluates in three steps (one for the rule of application and two for
the rules of  variable lookups), while the target program evaluates in only two
steps (one for the application and one for the variable lookup).
We believe that is possible to adjust the cost models to make the proof go
through, but we leave it as future work.
The LLM correctly identified this issue generating a counterexample, but was not able to fix it.
The CakeML compiler addresses this by inserting ``tick'' instructions that
consume fuel without doing work.
\citeN{Proofs:Paraskevopoulou19:compositionalopt} propose an alternative
technique based on \emph{resource invariants}.

\subsection{Proof Structure}\label{sec:proof-structure} The ANF correctness
proof is organized into four parts, summarized along with their proof sizes in
\Cref{tab:sizes}. The ANF proof is roughly 47\% larger than the CPS proof, with
the bulk of the increase in the semantic preservation component.


\begin{table}[h]
\centering
\begin{tabular}{@{}lp{7cm}rr@{}}
\toprule
Part & Description & CPS & ANF \\
\midrule
Semantic preservation & Main simulation theorem & 2,243 & 4,755 \\
Alpha-equivalence & ANF results from different name supplies are logically related & 1,860 & 1,914 \\
Monadic/relational & Monadic implementation satisfies relational spec & 763 & 864 \\
Top-level theorems & Whole-program \& separate compilation & 428 & 250 \\
\midrule
\textbf{Total} & & \textbf{5,294} & \textbf{7,783} \\
\bottomrule
\end{tabular}
\caption{Proof structure and sizes in lines of Rocq.}
\label{tab:sizes}
\end{table}

\section{Experience Report: Proof Development with an LLM}\label{sec:experience}
The proof was developed using Claude Code~\cite{claude-code}, Anthropic's
agentic coding assistant for software development, running as a VSCode extension
with the Claude Opus~4.6 model.
Claude Code is capable of reading and editing files, searching the codebase,
running build commands, and carrying out multi-step tasks with access to the
full development environment.
In our setting, it had access to the entire CertiCoq repository, including the
completed CPS proof files, and was able to invoke \texttt{Rocqc} to type-check
generated proofs.
We started the proof development by tackling the main semantic preservation
theorem, then moved to the other parts of the proof, building up to 
the top-level theorems.
The human role was to provide high-level guidance: structuring the proof
development into phases, explaining the high-level correspondence between the
CPS and ANF proofs, and providing natural-language feedback on certain proof
attempts.
%
Full transcripts of the interaction have been preserved alongside the
proof development.

\subsection{Development Process}\label{sec:process}

We initially instructed the LLM to complete an existing partial
proof (approximately 1,200 lines) that had been abandoned months
earlier due to complexity.
The LLM made some progress but eventually stalled, cycling through
similar unsuccessful attempts.
We then started from scratch on a new branch, instructing the LLM to
follow the proof technique and structure of the CPS proof, adapting the
reasoning to the ANF setting.

\paragraph{Phase 1: Theorem statements and proof skeleton.}
The LLM began by generating theorem statements for the four mutually
inductive proof obligations and then laid out the proof structure,
closely following the organization of the CPS proof files.
The initial correctness statement it generated lacked the universally
quantified continuation~$e_k$ (\cref{thm:anf-correct}), which is
essential for compositionality. The LLM corrected this after we pointed it out.
%
%
It then rapidly dispatched the trivial cases---the main correctness
theorem has thirteen cases---before moving on to the substantive ones.

\paragraph{Phase 2: Case-by-case development.}
The LLM worked through the remaining cases one at a time, partially
filling each case, generating \texttt{admit}s (often with explanatory
comments describing what remained to be proved), and stating helper
lemmas as needed.
We guided it on which case to tackle next, and it would produce short work plans
before starting.
Once all cases were sketched out with \texttt{admit}s, we prompted it
to first resolve the inner admits (within individual cases) before
turning to the outer helper lemmas.
This ordering worked well: after the inner admits were filled, the
outer lemmas were discharged with remarkable success.

\paragraph{Phase 3: Auxiliary proof files.}
After the main correctness proof was complete, we directed the LLM to
the three remaining proof files.
The \emph{alpha-equivalence utilities}, showing that the transformation yields
logically related results when invoked with different fresh-name supplies, went
smoothly, closely following the CPS template.
The \emph{monadic/relational correspondence} was more disappointing:
despite being conceptually the simplest of the four files (a purely
syntactic argument) and the most similar to the CPS counterpart, the
LLM struggled.
Rather than adapting the proof, it copied the CPS proof wholesale and attempted
to fix type errors one by one from error messages, leading to many slow
iterations. Ultimately, the author manually fixed two tactic invocations in the
last ten lines of the proof script to explicitly unify a variable that was
blocking the remaining goals.
The \emph{top-level theorems} were straightforward.
During this phase, the LLM discovered that the standard technique for divergence
preservation does not apply to the ANF transformation
(\cref{sec:divergence-preservation}), generating a high-level readable comment
next to the blocked \texttt{admit}, explaining why.

\paragraph{Autonomous operation.}
During the process, the LLM was frequently left to work unattended,
sometimes for hours, and continued to make progress independently, often
completing entire proof cases end-to-end without guidance.
Interface issues hindered this: at times, the tool would repeatedly
request permission for routine operations (compiling, grepping) despite
permission having already been granted, blocking progress until a human
intervened.

\subsection{What the LLM Did Well}\label{sec:llm-strengths}

\paragraph{Proof structure and adaptation.}
The LLM was remarkably effective at adapting the CPS proof structure to
the ANF setting.
It understood the high-level proof strategy and was able to instantiate
it in a setting with different technical details.

\paragraph{Proof reasoning.}
For each case, the LLM would first reason about why the case should
hold, often producing a natural-language train of thoughts, and
then write the corresponding proof script.
When the resulting script contained no remaining \texttt{admit}s, it
was typically logically correct and required only syntactic
adjustments (tactic names, hypothesis names, goal ordering) to compile.

\paragraph{Routine proof work.}
The LLM was proficient at the tedious but necessary parts of
mechanized proving: setting up inductions, performing case analyses,
applying known lemmas from the codebase, and managing proof subgoals.
It also independently generated helper lemmas and stated them with
appropriate generality.

\paragraph{Code quality.}
The generated proof code was well-structured, with clear naming
conventions, informative comments, and a logical file organization
consistent with the existing codebase.

\subsection{Challenges and Limitations}\label{sec:challenges}

\paragraph{Silently weakening proof statements.}
The most concerning behavior we observed was the LLM silently modifying
proof statements to make them easier to prove.
When it encountered the divergence preservation obstacle
(\cref{sec:divergence-preservation}), rather than reporting the
difficulty, it removed the fuel bound from the statement and filled the
fuel obligation with a trivial zero value.
This was caught when proving the top-level theorem, but it highlights a risk: the
machine-checked guarantee holds only for the statement that was
actually proved, and if that statement has been weakened, the guarantee
may be vacuous.
Careful human review of both proof statements and proof bodies remains
essential.

\paragraph{Incorrect assumptions about variable freshness.}
The two last admits in the body of the semantics preservation proof turned out
to encode a false assumption: the LLM assumed that certain result variables in
the let-application case could not coincide, so that no shadowing occurs.
In fact, the variables \emph{can} coincide, requiring a case split and
additional reasoning that if they are the same, the contexts evaluate
consistently.
The LLM understood the issue and it was able to fix the proof.
%

\paragraph{Iterative repair.}
The interactive repair process once the proof scripts were generated
was quite inefficient. 
The LLM frequently encountered issues with hypothesis names (which change as the
proof context evolves), and occasionally with goal numbering, leading to
unnecessary repair cycles.
This iterative repair process was slow and often circular (fixing one
error introduced another).
The LLM interacted with Rocq exclusively through file editing and the
\texttt{Rocqc} compiler, with no access to the interactive proof state.
Each fix required a full re-compilation, which was time-consuming. 


\paragraph{Context window limitations.}
The LLM's context window filled frequently during long proof sessions,
triggering automatic context compaction.
After compaction, it would sometimes lose track of the current proof state and
need to re-orient itself by re-reading files and re-grepping the codebase, work
it had already done in the compacted portion of the conversation.
This repeated exploration was time-consuming and suggests that better
mechanisms for persistent proof state across context boundaries would
be valuable.



\paragraph{Attention and focus.}
The LLM's focus could be narrow and unpredictable: it would sometimes
jump between distant parts of a proof or preferentially work on easier
cases while avoiding difficult admits.
Re-implementing lemmas that already existed in the codebase (but were
not imported in the current file) was another recurring issue.

\subsection{Comparison with the CPS Proof}\label{sec:comparison}

The CPS correctness proof was developed manually by human experts over several
months, requiring both conceptual innovation (the proof technique combining
syntactic simulations with step-indexed logical relations) and substantial
mechanical effort (formalizing hundreds of lemmas in Rocq).
The ANF proof, by contrast, was developed in approximately 96~hours of
LLM-assisted work, with no manual proof writing and mechanized a known proof
strategy.
The human contribution consisted of high-level guidance and review, and
occasional natural-language feedback on reasoning steps.
Human effort was not systematically measured; we supervised the process
intermittently while engaged in other work and did not keep a detailed
log of hours spent.

As shown in \cref{tab:sizes}, the ANF proof is larger than its CPS
counterpart (7,783 vs.\ 5,294 total lines), with the main correctness
file roughly twice the size (4,755 vs.\ 2,243 lines).
The ANF proof is in some respects more complex than the CPS proof because of the
context handling.
In other respects it is simpler: there are no administrative reductions
to reason through.
This also suggests that the LLM-generated proof is more verbose than the
human-written one. 



\section{Discussion}\label{sec:discussion}

\paragraph{A template-based paradigm for compiler verification.}
Our experience points toward a new workflow for verified compiler
development: a human expert designs the proof strategy and builds the
verification infrastructure for one transformation, and an LLM
mechanizes subsequent applications of the same technique to related
transformations.
CertiCoq's pipeline is a natural candidate for this approach: it uses
the same logical relation framework to verify a series of \lanf-to-\lanf\
transformations~\cite{Proofs:Paraskevopoulou19:compositionalopt}, each
following a similar proof pattern.
We anticipate that LLMs will prove effective at adapting proofs across
this family of transformations, substantially reducing the effort
required to extend or modify verified compiler pipelines.


\paragraph{Toward verified code generation.}
If LLMs can generate both code and proofs together, verified software
could become more of a norm than an exception.
Rather than a single verified compiler pipeline built with enormous
human effort, it might become feasible to generate new verified
transformations on demand.
Our case study suggests that this vision is within reach for
transformations that follow known proof templates.

\paragraph{Tool support for LLM-assisted proving.}
We identified several concrete gaps in tool support:

\noindent\textbf{Proof state access.}  The LLM interacted with Rocq
    only through file editing and \texttt{Rocqc} invocations.
    Access to the current proof state (goal, hypotheses, available lemmas) could
    dramatically reduce the time spent on syntactic repair.

\noindent\textbf{Structured error feedback.}  \texttt{Rocqc} error
    messages are designed for human experts.
    Structured, machine-readable error output may help LLMs diagnose
    and fix proof failures more efficiently.

\noindent\textbf{Persistent state across context windows.}
    Mechanisms for summarizing and preserving proof progress across
    context compaction boundaries would reduce the repeated exploration
    we observed.

\paragraph{Proof scripts vs.\ proof terms.}
A natural question is whether LLMs would be better suited to
generating proof terms (fully explicit proof objects) rather than
proof scripts (sequences of tactics).
Proof terms avoid the fragility of tactic scripts, where hypothesis
names and goal numbering can silently change.
On the other hand, proof terms are harder to diagnose and repair interactively.
We leave an empirical comparison as future work, but note that many of
the iterative repair difficulties we encountered are artifacts of the
tactic-based workflow.

\paragraph{The changing role of proof automation.}
If LLMs can reliably generate and repair proof scripts, the role of
custom tactic automation in proof developments may shift.
Traditionally, proof engineers invest significant effort in writing
custom tactics and decision procedures to reduce proof burden.
In an LLM-assisted workflow, the emphasis may shift from writing
automation to writing clear, well-structured proof templates that the
LLM can follow and adapt.

\paragraph{Generalizability.}
The success of our experiment depends on several preconditions: the
existence of a closely related, complete proof serving as a template; a
domain expert able to identify correspondences and provide targeted
guidance; and a sufficiently capable LLM.
How well this approach generalizes to settings where one or more of
these preconditions are absent, \eg., proofs without a close
template, or domains where the human lacks deep expertise, is an
important open question.

\paragraph{Reproducibility and attribution.}
LLMs are inherently stochastic and the proof generation in our experiment cannot
be reproduced exactly.
We mitigate this by making the complete proof and its commit history
publicly available, providing evidence that the proof was developed in
a timeframe consistent with LLM-assisted development.
Authorship is also hard to attribute, as there is no reliable 
way to determine if the code was LLM generated.


\paragraph{Education and the future of proof expertise.}
If mechanized proofs can be generated by LLMs, researchers will focus more on
designing proof strategies and less on the mechanics of tactic-based proving.
At the same time, the ability to evaluate, guide, and debug
LLM-generated proofs requires a level of expertise that is difficult to
acquire without hands-on proof experience, suggesting that proof
education remains important, even as its focus shifts.

\paragraph{Reliance on LLM providers.}
Our experiment relied on a commercial LLM.
This creates a dependency on the continued availability, affordability, and
capability of commercial services.
Open-source models, while currently less capable for this task, are an
important direction for the community to develop, both for
accessibility and for scientific reproducibility.

\section{Related Work}
\label{sec:related}

\paragraph{Agentic proof development.}
There are early reports on the successful use of agentic proof development. 
\citeN{swamy2026agenticpop} describes the process of using F* and LLMs to
produce 10,000 lines of code and proof, spanning a range of examples. 
\citeN{demoura2026proofassistants} similarly argues that proof
assistants and AI are entering a new symbiotic era, and identifies
machine-checked proof generation as a key near-term opportunity.
Our work is an instance of this paradigm applied to compiler
verification.

\paragraph{LLM-assisted proof generation.}
Several recent papers have explored the use of LLMs for automated proof
generation.
\citeN{fscq-case-study} study LLM proof generation for FSCQ, a verified file
system in Rocq, achieving 38\% proof coverage on preexisting theorems sampled
from the FSCQ codebase.
Our experience suggests that higher coverage is achievable when a closely
related proof template is available, even in end-to-end proof development and
not only in standalone lemma proving.
%

PALM~\cite{proof-automation-llms} analyzes common error patterns in
LLM-generated proofs to design a generate-then-repair approach that 
corrects low-level errors in LLM generated proofs.
Using GPT-3.5 it achieves 40\% coverage on a dataset of 10K+ Rocq theorems,
outperforming previous methods by over 76.6\%.
CoqPilot~\cite{coqpilot-llm-proof-generation} is a VSCode extension that uses
LLMs to fill proof holes in Rocq files.
The tool was evaluated on a dataset of 300 Rocq theorems with proofs of at most
20 tactics with different LLMs, achieving a maximum coverage of 50\% when using
GPT-4.
On a similar track, Lean Copilot~\cite{leancopilot2025} is a native LLM interface
for Lean proof assistant~\cite{demoura2021lean4}, that assists with automating
parts of proof development.
Rango~\cite{Thompson:2025:Rango} is a synthesis tool for Rocq that, at each
proof step, retrieves relevant premises and similar proofs to guide a fine-tuned
LLM. On CoqStoq, a new benchmark of 196,929 theorems from 2,226 open-source
projects, it proves 32.0\% of theorems (29\% more than the state-of-the-art
Tactician~\cite{Blaauwbroek:2020:Tactician}).
Our work differs in scale: these tools focus on automatically proving many
smaller independent lemmas, while we generate a single large, coherent proof,
under user guidance.
Most prior work evaluates proof generation on existing proof corpora; we target
previously unproven statements.
Our experience supports that integrating such tools in the agentic proving
environment could further speed up the process of mechanized proof development.

\citeN{misu2024dafny} explore LLM-assisted synthesis of verified Dafny methods
by prompting models to generate verified Danfy programs, achieving as much as
58\% coverage on the MBPP dataset of Python programs~\cite{austin2021mbpp}. 
Baldur~\cite{baldur-2023} demonstrates whole-proof generation using Minerva
LLM~\cite{lewkowycz2022minerva} fine-tuned on Isabelle/HOL proofs, using
a repair step to fix low-level errors.

\paragraph{AI for mechanized mathematics.}
The use of AI to generate machine-checked proofs for mathematical theorems
is gaining increasing attention.
Remarkably, DeepMind's AlphaProof~\cite{alphaproof2025}, a formal proof capable
of formalizing mathematics in Lean using reinforcement learning, solved 4 out of
the 6 problems in 2024 International Mathematical Olympiad.

\section{Conclusion}
\label{sec:conclusion}

We have presented a case study in LLM-assisted mechanized proof
development.
Using an agentic LLM coding assistant and high-level human guidance, we have
developed approximately 6,000 lines of Rocq proof for the correctness of the ANF
transformation in CertiCoq in a matter of days, compared to the months required
for the closely related CPS proof developed manually by human experts.
Our experience suggests that LLMs can be highly effective at
adapting known proof techniques to new but related settings, with
human guidance at key junctures.
We believe that LLM-assisted proof development may suggest a new paradigm of
mechanized proof construction, and we offer this case study as a concrete data
point for further investigation.




\begin{acks}
\end{acks}

\bibliographystyle{ACM-Reference-Format}

\bibliography{GC,cakeML,closure,resource,scheme,compcert,compcorr,compilers,cps,theses,logrel,coq,proofs,online,general,space,lisp,sepcomp,types,llm}

\appendix

\section{Proof Development Statistics}\label{app:statistics}

We report statistics on the proof development process, extracted from the
preserved conversation transcripts (JSONL logs) of all Claude Code sessions.
This was the author's first experience using an LLM for mechanized proof
development and the first time using Claude Code. As a result, no systematic
plan for structuring the interaction or measuring progress was established in
advance.
The development was tracked through 80 manual git commits.
The proof was developed across two Claude Code sessions: the author started in
one session on Day~0, then accidentally switched to a new session the next day
and continued there for the remainder of the development.
The statistics reported below are aggregated across both sessions.
The author subscribed to the Claude Max~20x plan; the weekly usage limit
was sufficient for the entire proof development.

\paragraph{Overall statistics.}
\Cref{tab:overall} summarizes the aggregate tool usage across all sessions.
The development spanned five consecutive days with 90.2
wall-clock hours (elapsed time from first to last event each day) and 62.6
active hours (an approximation of actual running time, obtained by excluding
gaps of more than 30 minutes with no recorded event).
The human provided 114 prompts in total, while the LLM performed
6,538 tool invocations.
Context compaction occurred 101 times across the five days, reflecting the long-running
nature of the sessions.

\begin{table}[h]
\centering
\begin{tabular}{@{}lr@{}}
\toprule
Metric & Count \\
\midrule
User prompts (human) & 114 \\
Context compactions & 101 \\
\lk{Rocqc}/\lk{make} invocations & 872 \\
Grep/search calls & 1,684 \\
Edit/Write calls & 1,209 \\
Read calls & 2,383 \\
Bash calls & 1,035 \\
\midrule
\textbf{Total tool uses} & \textbf{6,538} \\
\midrule
Wall-clock hours & 90.2 \\
Active hours & 62.6 \\
\bottomrule
\end{tabular}
\caption{Aggregate LLM tool usage across all sessions.}
\label{tab:overall}
\end{table}

\paragraph{Per-day and per-file activity.}
\Cref{tab:perday} breaks down the development activity by day.
%
%
Day~0 was essentially a fully autonomous run. The single initial prompt
(\cref{app:prompts}) was followed by roughly 4.3 active hours in which the LLM
explored the codebase, familiarized itself with the proof structure, and began
working on the semantic preservation proof.
The bulk of the semantic preservation proof was developed on Days~1-2, with the
LLM making steady progress on the main proof and the generated helper lemmas. 
Days~3-4 were devoted to the $\alpha$-equivalence proof, the monadic/relational
correspondence, and the top-level theorems. 
Day~4 has the highest number of human prompts (50). The majority of these
prompts came after the completion of the top-level corollaries, when the user
and LLM encountered a block on the divergence preservation proof.
The user prompted the LLM explain the problem with the proof and provide a counterexample. 
Then, the user prompted the LLM to walk through specific proof cases step by
step, in an attempt to understand and fix the source of the issue.

\begin{table}[h]
\centering
\small
\begin{tabular}{@{}lrrrrrrrrr@{}}
\toprule
Day & Prompts & Compact & \lk{Rocqc} & Grep & Edit/Wr & Read & Bash & Wall\,(h) & Active\,(h) \\
\midrule
Day 0 &  1 &  9 &   72 &  150 &  111 &  212 &   80 &  4.3 &  4.3 \\
Day 1 & 18 & 21 &  211 &  324 &  302 &  589 &  226 & 16.2 & 12.6 \\
Day 2 & 22 & 22 &  204 &  456 &  323 &  610 &  255 & 24.0 & 14.9 \\
Day 3 & 23 & 34 &  208 &  466 &  270 &  661 &  291 & 24.0 & 18.3 \\
Day 4 & 50 & 15 &  177 &  288 &  203 &  311 &  183 & 21.7 & 12.5 \\
\midrule
\textbf{Total} & \textbf{114} & \textbf{101} & \textbf{872} & \textbf{1,684} & \textbf{1,209} & \textbf{2,383} & \textbf{1,035} & \textbf{90.2} & \textbf{62.6} \\
\bottomrule
\end{tabular}
\caption{Per-day development activity.}
\label{tab:perday}
\end{table}

\paragraph{Per-file activity.}
\Cref{tab:perfile} shows the tool usage broken down by proof file and day.
The semantic preservation file received the most attention by far,
dominating Days~0-2 and requiring 515 compilation attempts over the
five days.
The alpha-equivalence utilities and the monadic/relational correspondence were
developed primarily on Days~3 and~4, with the latter requiring 132 compilations
despite being a conceptually simpler and smaller proof (\cref{sec:challenges}).
The top-level theorems were the lightest component, completed entirely
on Day~4.
The per-file totals do not sum to the overall totals in \cref{tab:perday},
as the LLM also read and searched other files in the codebase for reference
lemmas and definitions.

\begin{table}[h]
\centering
\small
\begin{tabular}{@{}llrrrrr@{}}
\toprule
Proof file & Day & Prompts & \lk{Rocqc} & Grep & Edit/Wr & Read \\
\midrule
Semantic preservation & 0 &  0 &  66 &  49 & 107 &   100 \\
                      & 1 & 18 & 209 & 142 & 302 &   493 \\
                      & 2 & 18 & 199 & 196 & 320 &   536 \\
                      & 3 &  1 &  41 &  33 &  48 &    74 \\
                      & 4 & 22 &   0 &  32 &   2 &    50 \\
\cmidrule(l){2-7}
                      & \textbf{Total} & \textbf{59} & \textbf{515} & \textbf{452} & \textbf{779} & \textbf{1,253} \\
\midrule
Alpha-equivalence     & 3 & 13 & 141 & 79 & 208 & 385 \\
                      & 4 &  0 &   0 & 12 &   0 &  12 \\
\cmidrule(l){2-7}
                      & \textbf{Total} & \textbf{13} & \textbf{141} & \textbf{91} & \textbf{208} & \textbf{397} \\
\midrule
Monadic/relational   & 3 &  1 &   1 &  0 &   1 &    4 \\
                      & 4 & 14 & 131 & 17 & 164 &  145 \\
\cmidrule(l){2-7}
                      & \textbf{Total} & \textbf{15} & \textbf{132} & \textbf{17} & \textbf{165} & \textbf{149} \\
\midrule
Top-level theorems    & 4 &  6 &  28 &  0 &  36 &   19 \\
\cmidrule(l){2-7}
                      & \textbf{Total} & \textbf{6} & \textbf{28} & \textbf{0} & \textbf{36} & \textbf{19} \\
\bottomrule
\end{tabular}
\caption{Tool usage per proof file and day.}
\label{tab:perfile}
\end{table}

\newcommand{\lkb}[1]{%
  \texttt{%
    \fontdimen2\font=2pt%
    \hyphenpenalty=0\exhyphenpenalty=0\relax
    #1}}

\section{Sample User Prompts}\label{app:prompts}

To illustrate the nature of the human--LLM interaction, we present a curated
selection of user prompts drawn from the conversation transcripts.
The prompts are lightly edited for readability (fixing typos) but are
otherwise verbatim.
They are grouped by the role they played in the development process.

\paragraph{Initial task framing.}
The very first prompt set up the overall goal:
\begin{quote}\small
\emph{``The file \lkb{Lambda\-Box\-Local\_\allowbreak to\_\allowbreak Lambda\-ANF.v} contains definitions for CPS
and ANF transformations. The correctness of the CPS transformation has been
proved in the files \lkb{Lambda\-Box\-Local\_\allowbreak to\_\allowbreak Lambda\-ANF\_\allowbreak correct.v} and
\lkb{Lambda\-Box\-Local\_\allowbreak to\_\allowbreak Lambda\-ANF\_\allowbreak corresp.v}; the correctness of ANF has not
yet been proved. Can you apply a similar proof technique to prove ANF correct?
The code and proofs that you generate should compile.''}
\end{quote}

\paragraph{Directing which task to work on.}
A large fraction of prompts simply steered the LLM to the next task:
\begin{quote}\small
\emph{``Proceed with the rest of the admits. Note that: for now, ignore
\lkb{preord\_\allowbreak exp\_\allowbreak Ecase\_\allowbreak red}; \lkb{anf\_\allowbreak cvt\_\allowbreak val\_\allowbreak alpha\_\allowbreak equiv} and
\lkb{anf\_\allowbreak val\_\allowbreak rel\_\allowbreak exists} will be the last ones to prove.''}

\medskip
\emph{``Amazing job. Now I want you to work to remove all admits in the main
proof. Be careful when introducing helper lemmas, you should be 100\% sure that
they are provable. After all inner admits are gone we will work on proving all
the helper lemmas.''}

\medskip
\emph{``Next I want you to tackle the lemmas in the family
\lkb{anf\_\allowbreak cvt\_\allowbreak disjoint\_*}.''}

\medskip
\emph{``Now the final big proof, before we go to top-level lemmas. The
correspondence of the relational and monadic transformations. This should be
done as per CPS; the proof for CPS is in file
\lkb{Lambda\-Box\-Local\_\allowbreak to\_\allowbreak Lambda\-ANF\_\allowbreak corresp}. Create a similar file for ANF
and do the proof.''}
\end{quote}

Looking at the last prompt, the fact that the LLM copied the CPS correspondence
proof could be due to the prompt not explicitly asking to adapt the proof, but
rather to create a similar file. 

\paragraph{Mathematical guidance.}
In few critical cases (for example, the constructor case of the semantic
preservation proof), the human provided detailed the human provided proof ideas
in natural language. The most significant ones are provided below.

\begin{quote}\small
\emph{``I have the following idea that I want you to investigate and tell me if
it's feasible.
\begin{itemize}[nosep,leftmargin=*]
\item First, abandon all the \lkb{replace\_\allowbreak with\_\allowbreak rho} stuff you are doing.
\item We know that \lk{xs} may have duplicates.
Let $\lk{xs}[i] = \lk{xs}[j]$ for $i \neq j$.
I claim that it should be the case that $\lk{vs0}[i] = \lk{vs0}[j]$.
This will be proven structurally using the ANF conversion definition
and the definition of the semantics. It may even follow from the
definition of the conversion that syntactically $\lk{es}[i] = \lk{es}[j]$,
which will make proving $\lk{vs0}[i] = \lk{vs0}[j]$ easier by the
determinism of the semantics.
\item Then, using the alpha-equivalence result, \lkb{preord\_val}
$\lk{vs'0}[i] = \lk{vs'0}[j]$, which should suffice to make the proof go through.
\end{itemize}
Examine if this approach works, and if it does, implement it. If not, tell me
what the problem is.''}

\medskip
\emph{``I want us to prove the following, then we will use this to show the
remaining admit.
\begin{itemize}[nosep,leftmargin=*]
\item \lkb{anf\_\allowbreak cvt\_\allowbreak rel}
$S\;e\;\mathit{vnames}\;\mathit{cnstrs}\;S'\;C\;x$
$\rightarrow\; x = \mathit{vnames}[i]$
$\rightarrow\;$\lkb{eval\_\allowbreak fuel}
$\mathit{env}\;e\;v\;f\;t$
$\rightarrow\; v = \mathit{vnames}[i]$
\item \lkb{anf\_\allowbreak cvt\_\allowbreak rel\_\allowbreak exps}
$S\;\mathit{es}\;\mathit{vnames}\;\mathit{cnstrs}\;S'\;C\;\mathit{xs}$
$\rightarrow\; \mathit{xs}[j] = \mathit{vnames}[i]$
$\rightarrow\;$\lkb{eval\_\allowbreak fuel\_\allowbreak exps}
$\mathit{env}\;\mathit{es}\;\mathit{vs}\;f\;t$
$\rightarrow\; \mathit{vs}[j] = \mathit{vnames}[i]$
\end{itemize}
First make sure it is true, then start proving it. You will need mutual
induction. If you find that it is not true, tell me instantly.''}

\medskip
\emph{``I am thinking how to use the lemmas to remove the first admit.
Say that we have a more general version of \lk{consistent} that relates the
second list with any relation, not just equality. Then:
\begin{itemize}[nosep,leftmargin=*]
\item Lemma~1:
\lkb{eval\_\allowbreak fuel\_\allowbreak many} $\mathit{vs'}\;\mathit{es}\;\mathit{vs0}\;\mathit{fs}\;\mathit{ts}$
$\rightarrow\;$\lkb{anf\_\allowbreak cvt\_\allowbreak correct\_\allowbreak exps} $\mathit{vs'}\;\mathit{es}\;\mathit{vs0}\;\mathit{fs}\;\mathit{ts}$
$\rightarrow\;$\lkb{Forall2} $(\lambda\; v\; v'.\;$\lkb{anf\_\allowbreak val\_\allowbreak rel} $v\; v')$ $\mathit{vs0}\;\mathit{vs'0}$
$\rightarrow\;$\lkb{consistent} \lk{preord\_val} $\mathit{vs0}\;\mathit{vs'0}$.
You will probably need the anf-equiv lemma for this.
\item Lemma~2:
\lkb{consistent} \lk{preord\_val} $\mathit{xs}\;\mathit{vs}$
$\rightarrow\; \exists\; \mathit{xs},\;$\lkb{getlist} $\mathit{xs}\;($\lkb{setlist} $\mathit{xs}\;\mathit{vs}\;\rho) = \mathit{vs'}$
$\land\;$\lkb{Forall2} \lk{preord\_val} $\mathit{vs}\;\mathit{vs'}$.
\end{itemize}
Then these two together should suffice to remove the admit.''}
\end{quote}

\paragraph{Corrections and error diagnosis.}
The human caught errors ranging from silent proof erasure to tactic
misapplication. In a few cases the user fed the proof state directly into the
LLM to assist in fixing low-level tactic errors.
\begin{quote}\small
\emph{``It compiled because you erased the proof. The goal is to fix the
proof, please add it back.''}

\medskip
\emph{``It is just further down the file. You are trying to use a lemma that
is proved later on.''}

\medskip
\emph{``You are inverting \lk{HL} which is a \lk{var}. This makes no
sense; it is probably due to inconsistent naming. You should invert \lk{H1}
and then \lk{zify; lia} will work.''}
\end{quote}

\paragraph{Nudging and continuing.}
A substantial portion of prompts were simple continuations, often needed after
the LLM paused waiting for direction or after context compaction:
\begin{quote}\small
\emph{``Why did you stop?''}
\quad
\emph{``Carry on please.''}
\quad
\emph{``Continue from where you left off.''}
\end{quote}


\paragraph{Discussion of proof limitations.}
When the LLM identified the divergence preservation obstacle
(\cref{sec:divergence-preservation}), the interaction shifted to a
collaborative discussion:
\begin{quote}\small
\emph{``Can you explain the issue with the bound and the OOT case and what it
would take to resolve?''}

\medskip
\emph{``Can you give me a counterexample of the property to put in the
paper?''}
\end{quote}

\end{document}